\documentclass[conference]{IEEEtran}
\usepackage[cmex10]{amsmath}
\usepackage{amsthm,amssymb,amsfonts,amscd}
\usepackage{algorithmic}
\usepackage{array}
\usepackage{booktabs}
\usepackage[font=footnotesize]{caption}
\usepackage{graphicx}
\usepackage{enumerate}
\usepackage[tight,footnotesize]{subfigure}
\theoremstyle{definition}
\newtheorem{theorem}{Theorem}

\newtheorem{lemma}[theorem]{Lemma}
\newtheorem{corollary}[theorem]{Corollary}
\newtheorem{property}[theorem]{Property}
\theoremstyle{remark}

\begin{document}
%
\title{The Least Degraded and the Least Upgraded Channel with respect to a Channel Family}
\author{\IEEEauthorblockN{Wei Liu, S.~Hamed Hassani, and R\"{u}diger Urbanke}
\IEEEauthorblockA{School of Computer and Communication Sciences\\
EPFL, Switzerland\\
Email: \{wei.liu, seyedhamed.hassani, rudiger.urbanke\}@\;epfl.ch}}

\maketitle

\begin{abstract}
Given a {\em family} of binary-input memoryless output-symmetric
(BMS) channels having a {\em fixed} capacity, we derive the BMS
channel having the highest (resp. lowest) capacity among all channels
that are degraded (resp. upgraded) with respect to the whole family.
We give an explicit characterization of this channel as well as an
explicit formula for the capacity of this channel.
\end{abstract}

\IEEEpeerreviewmaketitle

\section{Introduction}
Channel ordering plays an important role in analyzing the asymptotic
behavior of iterative coding systems over binary-input memoryless
output-symmetric (BMS) channels. We say that two BMS
channels are {\em ordered} by a real-valued parameter, if the
``worse'' channel has a larger parameter than the ``better'' channel.
E.g., we might consider as a real-valued parameter the bit-error
probability, the capacity, or perhaps the Battacharyya parameter.

One particularly useful way of imposing an order is to consider the
notion of {\em degradation} \cite{ru08}. The main reason that this order is
so useful is that many important functionals and kernels (defined
later) associated with BMS channels either preserve or reserve the
ordering induced by degradation. In addition, the order of degradation
is preserved under many natural operations one might perform on channels.

Instead of considering a single BMS channel, we are interested in
a family of BMS channels sharing certain properties. The
cardinality of the family can be either finite or infinite.
A particularly important example for our purpose is the family of
BMS channels having fixed capacity $c$. This family is denoted by
$\{\textrm{BMS}(c)\}$. The question we address is the following.
Given the family of all BMS channels of capacity $c$, can we find
the ``best'' degraded channel and the ``worst'' upgraded channel
with respect to this family. I.e., can we find the channel that
is degraded with respect to all members in the family and has the highest
capacity, as well as the channel that is upgraded with respect to
all members of the family and has the lowest capacity.

Determining these channels is both natural and fundamental. To
mention just one possible application, assume that we want to
construct a polar code \cite{arikan09} which works for all channels of a particular
capacity. If we can find a channel which is degraded with respect
to all elements in the family then a polar code designed for this
channel will work a fortiori for all elements of the original family.
We then say that this polar code is {\em universal}. This simple
bound is somewhat too crude to construct a good universal polar
code. But we can apply the same argument not to the original family
but after a few steps of the polarization process. The more steps
we perform the better the overall code will be (i.e., it will have
higher capacity) and the closer we will get to an optimal construction.
In what follows we will not be concerned with applications but we
will only be interested in describing these two extreme channels
given a channel family. As we will see, we will be able to give an
explicit answer to the question.

\subsection{Channel Model}
A binary-input memoryless output-symmetric (BMS) channel has an
input $X$ taking values over an input alphabet $\mathcal{X}$, an
output $Y$ taking values over an output alphabet $\mathcal{Y}$,
and a transition probability $p_{Y|X}(y|x)$ for each $x \in
\mathcal{X}$ and $y \in \mathcal{Y}$. Throughout the paper we take the input
alphabet $\mathcal{X}$ to be $\{\pm 1\}$ and the output alphabet
$\mathcal{Y}$ to be a subset of $\bar{R}\triangleq [-\infty,+\infty]$.

Given a BMS channel $a$, an alternative description of the channel
is the $L$-distribution of the random variable
\begin{equation*}
L = l(Y) \triangleq \ln \frac{p_{Y|X}(Y|1)}{p_{Y|X}(Y|-1)}.
\end{equation*}
A closely related random variable is
\begin{equation*}
D = d(Y) \triangleq \tanh\left(\frac{l(Y)}{2}\right) = \frac{1-\textsf{e}^{-l(Y)}}{1+\textsf{e}^{-l(Y)}},
\end{equation*}
whose distribution is called a $D$-distribution, conditioned on
$X=1$. Moreover, denote by $|D|$ the absolute value of the random
variable $D$. Then, the distribution of $|D|$ is termed a
$|D|$-distribution and the associated density is called a $|D|$-density.
Given a BMS channel $a$, we denote by $|\mathfrak{A}|$and
$|\mathfrak{a}|$ the associated $|D|$-distribution and $|D|$-density,
respectively. The $|D|$-density of a {\em discrete} BMS channels
is of the form
\begin{equation}\label{equ-abs-d-density}
|\mathfrak{a}|(x) = \sum_{i=1}^n \alpha_i \delta(x - x_i),
\end{equation}
where $0 \leq \alpha_i \leq 1$, $\sum_{i=1}^n \alpha_i = 1$,
and $0 \leq x_1 < x_2 < \cdot\cdot\cdot < x_n \leq 1$;
while $\delta(\cdot)$ is the (Dirac) delta function and $n \in
\mathbb{N}_{+}$ is finite or countably infinite.

\subsection{Functionals of Densities}
The {\em capacity} of a BMS channel can be defined as a linear
functional of its $|D|$-density \cite{ru08}. Formally, the capacity
$\mathrm{C}(|\mathfrak{a}|)$ in bits per channel use of a BMS
channel with $D$-density $|\mathfrak{a}|$ is defined as
\begin{equation*}
\mathrm{C}(|\mathfrak{a}|) \triangleq \int_0^1 |a|(x) \left(1 - \textsf{h}(x)\right) \mathrm{d}x,
\end{equation*}
where $\textsf{h}(x) \triangleq h_2((1-x)/2)$, for $x \in [0,1]$,
and $h_2(x) \triangleq -x\log_2(x) - (1-x)\log_2(1-x)$ is the {\em binary entropy function}.
The {\em entropy functional} is defined as $\mathrm{H}(|\mathfrak{a}|) \triangleq 1 -
\mathrm{C}(|\mathfrak{a}|)$, and the entropy of a discrete
BMS channel is $\sum_{i=1}^n \alpha_i \textsf{h}(x_i)$.

Another important functional is the {\em Battacharyya parameter}
associated with the $|D|$-density $|\mathfrak{a}|$. It is defined as
\begin{equation*}
\mathfrak{B}(|\mathfrak{a}|) \triangleq \int_0^1 |\mathfrak{a}|(x)\sqrt{1-x^2}\mathrm{d}x.
\end{equation*}
If a BMS channel has a $|D|$-density of the form in
(\ref{equ-abs-d-density}), the Battacharyya parameter is given by
$\sum_{i=1}^n\alpha_i\sqrt{1-x_i ^2}$.

\subsection{Degradedness}
Denote by $p_{Y|X}(y|x)$ and $p_{Z|X}(z|x)$ the transition probabilities
of two BMS channels $a$ and $a'$, respectively. Let $\mathcal{X}$ be
the common input alphabet, and denote by $\mathcal{Y}$ and $\mathcal{Z}$
the output alphabets, respectively. We say that $a'$ is
{\em (stochastically) degraded} with respect to $a$ if there
exists a memoryless channel with transition probability $p_{Z|Y}(z|y)$
such that
\begin{equation*}
p_{Z|X}(z|x) = \sum_{y \in \mathcal{Y}} p_{Z|Y}(z|y) p_{Y|X}(y|x),
\end{equation*}
for all $x \in \mathcal{X}$ and all $z \in \mathcal{Z}$. Conversely,
the channel $a$ is said to be {\em (stochastically) upgraded}
with respect to the channel $a'$. The following lemma gives an
equivalent characterization of degradedness and upgradedness
\cite{ru08}.
\begin{lemma}\label{lemma-degrad}
Consider two BMS channels $a$ and $a'$ with the corresponding $|D|$-distributions $|\mathfrak{A}|$ and $|\mathfrak{A}'|$, respectively. The following statements are equivalent:
\begin{enumerate}[(i)]
\item $a'$ is degraded with respect to $a$;
\item $\int_0^1 f(x)\mathrm{d}|\mathfrak{A}|(x) \leq \int_0^1 f(x) \mathrm{d}|\mathfrak{A}'|(x)$, for all $f$ that are non-increasing and convex-$\cap$ on $[0,1]$;
\item $\int_z^1 |\mathfrak{A}|(x) \mathrm{d}x \leq \int_z^1 |\mathfrak{A}'|(x) \mathrm{d}x$, for all $z \in [0,1]$.
\end{enumerate}
\end{lemma}

\section{Least Degraded Channel}
Now recall that the entropy functional $\mathrm{H}(|\mathfrak{a}|)$
associated with the $|D|$-density $|\mathfrak{a}|$ has $h_2((1-x)/2)$
as its kernel in the $|D|$-domain. This kernel is non-increasing
and convex-$\cap$ on $[0,1]$, and therefore condition (ii) in Lemma
\ref{lemma-degrad} shows that a degraded channel has a larger
entropy, and thus lower capacity, than the original channel.
Similarly, a degraded channel has a larger probability of error and
larger Battacharyya parameter than the original channel.
\begin{property}
Note that there are many equivalent definitions for the {\em least
degraded} channel. One particularly insightful characterization is
that the least degraded channel is the unique channel which is
degraded with respect to all elements of the family $\{\textrm{BMS}(c)\}$
and is upgraded with respect to any other such channel.
Therefore it has the highest capacity, lowest error probability and lowest Battacharyya parameter.
\end{property}

The integral of $|\mathfrak{A}|(x)$ from $z$ to $1$,
as suggested by Lemma \ref{lemma-degrad},
is important for characterizing degradedness among BMS channels.
Therefore, for a BMS channel $a$, define
\begin{equation}
\Lambda_a(z) \triangleq \int_z^1 |\mathfrak{A}|(x) \mathrm{d}x.
\end{equation}
Note that $\Lambda_a(z)$ is decreasing on $[0,1]$,
and since $|\mathfrak{A}|(x)$ is increasing in $x$,
it follows that $\Lambda_a(z)$ is a convex-$\cap$ function of $z$.
Moreover, if the BMS channel $a$ is discrete,
the function $\Lambda_a(z)$ has a simpler form given by
\begin{equation}\label{equ-lambda-form}
\Lambda_a(z) = \sum_{i=1}^n \alpha_i \left(1 - \max\left\{z, x_i\right\}\right).
\end{equation}

Condition (iii) of Lemma \ref{lemma-degrad} shows that,
in order to find the least degraded channel with respect to $\{\textrm{BMS}(c)\}$,
we take the maximum value of $\Lambda_a(z)$ among all BMS channels $a \in \{\textrm{BMS}(c)\}$,
for each fixed $z \in [0,1]$. {E.g.,}~define
\begin{equation}
\overline{\Lambda}(z) \triangleq \max \left\{\Lambda_a(z): a \in \{\textrm{BMS}(c)\}\right\},
\end{equation}
for every $z\in[0,1]$. Taking the convex-$\cap$ envelope of
$\overline{\Lambda}(z)$ then characterizes the desired channel. The following
theorem characterizes $\overline{\Lambda}(z)$ exactly.
\begin{theorem}\label{theorem-overline-Lambda} Consider the family
of BMS channels $\{\textrm{BMS}(c)\}$ of capacity $c$, $0 < c < 1$.
Then, \begin{equation*} \overline{\Lambda}(z) = \begin{cases}
\cfrac{(1-c)(1-z)}{\textsf{h}(z)} & \text{if } z \in
[0,1-2\epsilon_{\textsf{bsc}}),\\ 1-z & \text{if } z \in
[1-2\epsilon_{\textsf{bsc}},1], \end{cases} \end{equation*} where
$\epsilon_{\textsf{bsc}} \in (0,\frac{1}{2})$ is the solution of
$1-h_2(\epsilon_{\textsf{bsc}})=c$.  \end{theorem} \begin{proof}
First, it is clear that $\sum_{i=1}^n\alpha_i(1-\max\{z,x_i\}) \leq
\sum_{i=1}^n\alpha_i(1-z) = 1-z$, and for $z \in [1 -
2\epsilon_{\textsf{bsc}},1]$ this value, $1-z$, is achieved
if the underlying BMS channel is the BSC with crossover probability
equal to $\epsilon_{\textsf{bsc}}$.

Now for any fixed $z \in [0,1 - 2\epsilon_{\textsf{bsc}})$, assume
that $\overline{\Lambda}(z)$ is achieved by the BMS channel $d$,
{e.g.,}~$\overline{\Lambda}(z)=\Lambda_d(z)$. We claim
that $d$ does not have any probability mass in the interval $(z,1)$.
We show this by contradiction. Suppose that there
exists a probability mass $\alpha_0$ at $x_0 \in (z,1)$. Define
$\rho \triangleq \textsf{h}(x_0) / \textsf{h}(z)$. It is clear that
$\rho \in (0,1)$, since the function $\textsf{h}(\cdot)$ is decreasing
and convex-$\cap$ on $[0,1]$. The definition of $\rho$ gives
\begin{equation}\label{equ-h-mass}
\alpha_0 \textsf{h}(x_0) = \alpha_0\rho \textsf{h}(z) + \alpha_0(1-\rho) \textsf{h}(1),
\end{equation}
which means that, instead of putting the single probability mass $\alpha_0$ at $x_0$,
we can split $\alpha_0$ into two masses, $\alpha_0\rho$ and $\alpha_0(1-\rho)$,
and put these two masses at $z$ and $1$, respectively,
without changing the entropy and thus the capacity.
Canceling $\alpha_0$ on both sides of (\ref{equ-h-mass}) and using the fact
that $\textsf{h}(\cdot)$ is decreasing and convex-$\cap$ on $[0,1]$, we have
\begin{equation*}
\textsf{h}(x_0) < \textsf{h}(\rho z + 1 - \rho),
\end{equation*}
which is equivalent to $1-x_0 < \rho(1-z)$. Now notice that, since $z < x_0 < 1$,
the term corresponding to $x_0$ that contributes to $\Lambda_d(z)$ is equal to $\alpha_0(1-x_0)$. But
\begin{equation*}
\alpha_0(1-x_0) < \alpha_0\rho(1-z) = \alpha_0\rho(1-z) + \alpha_0(1-\rho)(1-1),
\end{equation*}
which means that the value of $\Lambda_d(z)$ can be increased by
the splitting operation mentioned above. This, however, contradicts
the assumption that $\Lambda_d(z)$ is the maximum value at $z$. If
there exists some probability masses on the interval $[0,z)$, we
can add these masses to the probability mass at $z$ and delete the
original masses, without changing $\Lambda_d(z)$ and without violating
the entropy (or capacity) constraint. Thus, the channel $d$ has
probability masses at points $z$ and $1$ only, and the probability
mass at $z$ is equal to $(1-c)/\textsf{h}(z)$, completing the proof.
\end{proof}

Recall that the  function $\Lambda$ associated to a BMS channel is convex-$\cap$
on $[0,1]$. Thus, taking the convex-$\cap$ envelope of
$\overline{\Lambda}(z)$ gives the $\Lambda$ function, call it
$\Lambda^{*}(z)$, of the least degraded channel with respect to the
whole channel family $\{\textrm{BMS}(c)\}$.
\begin{corollary}
The least degraded channel is characterized by
\begin{equation*}
\Lambda^{*}(z) = \begin{cases}
1-c-\cfrac{1-c-2\epsilon_{\textsf{bsc}}}{1-2\epsilon_{\textsf{bsc}}}\,z & \text{if } z \in [0,1-2\epsilon_{\textsf{bsc}}),\\
1-z & \text{if } z \in [1 - 2\epsilon_{\textsf{bsc}},1].
\end{cases}
\end{equation*}
\end{corollary}
Once having this characterization, we can derive the
exact formula of the capacity of the least degraded channel.
\begin{theorem}\label{theorem-de-c}
Given a family of BMS channels of capacity $c$,
the capacity in bits per channel use of the least degraded channel is given by
\begin{equation}
\mathrm{C}^{*} = \cfrac{c^2}{1-2\epsilon_{\textsf{bsc}}}.
\end{equation}
\end{theorem}
\begin{proof}
Recall that the entropy can be expressed in the following alternative form,
\begin{equation}\label{equ-H-alt}
\mathrm{H}(|\mathfrak{a}|) = \int_0^1\cfrac{1}{\ln2(1-z^2)}\left(\int_z^1|\mathfrak{A}|(x)\right)\mathrm{d}z.
\end{equation}
Inserting the formula for $\Lambda^{*}(z)$ into (\ref{equ-H-alt}) and integrating over $z$ gives the desired result.
\end{proof}

\section{Least Upgraded Channel}
\begin{property}
The {\em least upgraded} channel is the unique channel which is
upgraded with respect to to all elements of the family $\{\textrm{BMS}(c)\}$
and is degraded with respect to any other such channel. Therefore this channel has
the lowest capacity, highest error probability and highest Battacharyya parameter.
\end{property}
In order to specify the least upgraded channel, we first notice the following lemma.
\begin{lemma}\label{lemma-min-pos}
For any BMS channel associated with the $|D|$-density of the form in (\ref{equ-abs-d-density})
and having capacity $c$, $0 < c < 1$, we have $1-2\epsilon_{\textsf{bsc}} \leq x_n \leq 1$.
\end{lemma}
\begin{proof}
Assume on the contrary that $x_n < 1-2\epsilon_{\textsf{bsc}}$.
Then, the monotonicity property of $\textsf{h}(\cdot)$ gives
\begin{equation*}
\textsf{h}(x_n) > \textsf{h}(1-2\epsilon_{\textsf{bsc}}) = 1-c,
\end{equation*}
which is equivalent to $\sum_{i=1}^n \alpha_i \textsf{h}(x_i) > 1-c$,
contradicting the assumption that the entropy is $1-c$.
\end{proof}

Now, for each fixed $z \in [0,1]$, we take the minimum value of $\Lambda_a(z)$
among all BMS channels $a \in \{\textrm{BMS}(c)\}$, {e.g.,}~define
\begin{equation}
\underline{\Lambda}(z) \triangleq \min \left\{\Lambda_a(z): a \in \{\textrm{BMS}(c)\}\right\}.
\end{equation}
Suppose that, for any fixed $z \in [0,1]$,
the minimum value is achieved by the channel $u$,
{i.e.,}~$\underline{\Lambda}(z)=\Lambda_u(z)$. For channel $u$,
the number of probability masses on the interval $[z,1]$ is characterized by the following lemma.
\begin{lemma}\label{lemma-mass-on-z-1}
The channel $u$, which achieves the minimum value $\underline{\Lambda}(z)$ at $z$,
has at most one probability mass on $[z,1]$.
\end{lemma}
\begin{proof}
Assume there are two probability masses $\alpha_{-}$ and $\alpha_{+}$
at $x_{-} \in [z,1]$ and $x_{+} \in [z,1]$, respectively.
Then, there exists a point $\tilde{x}$ on $[x_{-},x_{+}]$ such that
\begin{equation}\label{equ-merge}
\alpha_{-}\textsf{h}(x_{-}) + \alpha_{+}\textsf{h}(x_{+}) = (\alpha_{-}+\alpha_{+})\textsf{h}(\tilde{x}).
\end{equation}
Dividing both sides of (\ref{equ-merge}) by ($\alpha_{-}+\alpha_{+}$)
and using the convex-$\cap$ property of $\textsf{h}(\cdot)$ gives
\begin{equation*}
\textsf{h}(\tilde{x}) < \textsf{h}\left(\frac{\alpha_{-}}{\alpha_{-}+\alpha_{+}}x_{-}+\frac{\alpha_{+}}{\alpha_{-}+\alpha_{+}}x_{+}\right),
\end{equation*}
which means that $\alpha_{-}x_{-}+\alpha_{+}x_{+} < (\alpha_{-}+\alpha_{+}) \tilde{x}$.
Since the terms corresponding to $x_{-}$ and $x_{+}$ that contribute to $\Lambda_u(z)$ is equal to $\alpha_{-}(1-x_{-})+\alpha_{+}(1-x_{+})$, we have
\begin{align*}
&\alpha_{-}(1-x_{-})+\alpha_{+}(1-x_{+}) \\
&= (\alpha_{-}+\alpha_{+})\left(1-\frac{\alpha_{-}}{\alpha_{-}+\alpha_{+}}x_{-} - \frac{\alpha_{+}}{\alpha_{-}+\alpha_{+}}x_{+}\right)\\
&> (\alpha_{-}+\alpha_{+})(1-\tilde{x}),
\end{align*}
which shows that by combining two probability masses $\alpha_{-}$ and $\alpha_{+}$ into a single probability mass $(\alpha_{-}+\alpha_{+})$ at position $\tilde{x}$, the value of $\Lambda_u(z)$ is decreased, contradicting the assumption that $\Lambda_u(z)$ is minimal at point $z$.
\end{proof}

Lemma \ref{lemma-min-pos} and Lemma \ref{lemma-mass-on-z-1} show that the channel $u$ has at most two probability masses on the interval $[0,1]$. If it has only one probability mass, it is in fact a BSC. Otherwise, in general, there are two probability masses, call them $\gamma$ and $1-\gamma$, on the intervals $[0,z]$ and $[1-2\epsilon_{\textsf{bsc}},1]$, respectively. Denote by $\check{x}$ and $\hat{x}$ the positions of these two masses, respectively.
\begin{lemma}\label{lemma-0-1}
Either $\check{x} = 0$ holds or $\hat{x} = 1$ holds, or $\check{x} = 0$ and $\hat{x} = 1$ hold simultaneously.
\end{lemma}
\begin{proof}
Suppose on the contrary that $\check{x} \neq 0$ and $\hat{x} \neq 1$. Then, consider decreasing $\check{x}$ by $\zeta$, where $\zeta \in \mathbb{R}_{+}$ is sufficiently small. Assume that $\hat{x}$ is increased by $\delta$ correspondingly, where $\delta \in \mathbb{R}_{+}$ is sufficiently small. First-order Taylor expansion gives
\begin{align*}
\textsf{h}(\check{x}-\zeta) &= \textsf{h}(\check{x}) - \zeta \textsf{h}'(\check{x}) + \mathcal{O}(\zeta^2),\\
\textsf{h}(\hat{x}+\delta) &= \textsf{h}(\hat{x}) + \delta \textsf{h}'(\hat{x}) + \mathcal{O}(\delta^2).
\end{align*}
Eliminating second and higher order terms results in
\begin{align*}
\textsf{h}(\check{x}) &\doteq  \textsf{h}(\check{x}-\zeta) + \zeta \textsf{h}'(\check{x}),\\
\textsf{h}(\hat{x}) &\doteq \textsf{h}(\hat{x}+\delta) - \delta \textsf{h}'(\hat{x}).
\end{align*}
Multiplying the above two equations by $\gamma$ and $1-\gamma$, respectively, and rearranging terms give
\begin{align*}
\gamma \textsf{h}(\check{x}) &\doteq  \gamma \left(\textsf{h}(\check{x}-\zeta) + \frac{\zeta \textsf{h}'(\check{x})}{\textsf{h}(\check{x}-\zeta)} \textsf{h}(\check{x}-\zeta)\right),\\
(1-\gamma) \textsf{h}(\hat{x}) &\doteq (1-\gamma)\left(\textsf{h}(\hat{x}+\delta) - \frac{\delta \textsf{h}'(\hat{x})}{\textsf{h}(\hat{x}+\delta)} \textsf{h}(\hat{x}+\delta)\right).
\end{align*}
Now, for any sufficiently small $\zeta > 0$, one can pick $\delta > 0$ small enough, such that
\begin{equation}
- \cfrac{\gamma \zeta \textsf{h}'(\check{x})}{\textsf{h}(\check{x}-\zeta)} = - \cfrac{(1-\gamma) \delta \textsf{h}'(\hat{x})}{\textsf{h}(\hat{x}+\delta)} \triangleq \gamma_{0},
\end{equation}
and $\gamma_{0} > 0$, and we then have
\begin{align*}
\gamma \textsf{h}(\check{x}) &\doteq (\gamma - \gamma_{0})\textsf{h}(\check{x}-\zeta),\\
(1-\gamma) \textsf{h}(\hat{x}) &\doteq (1-\gamma+\gamma_{0}) \textsf{h}(\hat{x}+\delta),
\end{align*}
which means that by deleting $\gamma_0$ from $\gamma$ at position $\check{x}-\zeta$ and adding $\gamma_0$ to $1-\gamma$ at position $\hat{x}+\delta$, one can keep the entropy constraint satisfied. Now, denote by $\Lambda_{u'}(z)$ the result after the above operations. Then, we can obtain
\begin{align*}
&\Lambda_{u'}(z) - \Lambda_{u}(z)\\
&= (\gamma-\gamma_0)(1-z) + (1-\gamma+\gamma_0)(1-\hat{x}-\delta)\\
&\quad \quad - \gamma(1-z) - (1-\gamma)(1-\hat{x})\\
&= \gamma_0(z-\hat{x})-(1-\gamma+\gamma_0)\delta < 0,
\end{align*}
which shows that $\Lambda_u(z)$ is not minimal at $z$, contradicting the assumption that the channel $u$ achieves the minimum value $\underline{\Lambda}(z)$ at $z$.
\end{proof}

Lemma \ref{lemma-0-1} shows that, for the channel $u$ achieving $\underline{\Lambda}(z)$ at $z$, its probability masses and in particular their associated positions cannot be arbitrary. Indeed, only the following three cases are possible.
\begin{enumerate}[(i)]
\item There is only one probability mass on the interval $[0,1]$. Then the channel $u$ is in fact a BSC such that there is a probability mass 1 at position $1-2\epsilon_{\textsf{bsc}}$, and $\Lambda_{\textsf{bsc}}(z) = 2\epsilon_{\textsf{bsc}}$ if $z \in [0,1-2\epsilon_{\textsf{bsc}})$, while $\Lambda_{\textsf{bsc}}(z) = 1-z$ if $z \in [1-2\epsilon_{\textsf{bsc}},1]$.
\item There are two probability masses $\gamma$ and $1-\gamma$ at positions $\check{x}$ and $\hat{x}$, respectively. Particularly, $\check{x}=0$, and $\hat{x} \in [1-2\epsilon_{\textsf{bsc}},1)$ and $\hat{x} \geq z$. In this case, we have $\Lambda_u(z)=\gamma(1-z)+(1-\gamma)(1-\hat{x})$, and the entropy constraint is $\gamma + (1-\gamma)\textsf{h}(\hat{x}) = 1-c$. Notice that $\hat{x}$ and $\gamma$ are parameters that should be optimized. Denote by $\Lambda_{\textsf{opt}}(z)$ the corresponding optimal value.
\item There are two probability masses $1-c$ and $c$ at positions $\check{x}=0$ and $\hat{x}=1$; namely, the channel $u$ is a BEC, and $\Lambda_{\textsf{bec}} = (1-c)(1-z)$ for $z \in [0,1]$.
\end{enumerate}

In order to find the minimum value $\underline{\Lambda}(z)$ at each point $z \in [0,1]$, now it suffices to compare the point-wise results of the above three cases. Case (ii) above is a non-trivial case, since $\hat{x}$ and $\gamma$ are unknown parameters. However, the next lemma characterizes the optimal solutions of these parameters.
\begin{lemma}\label{lemma-optimization}
The optimization problem
\begin{equation}
\begin{aligned}
& \underset{\gamma,\hat{x}}{\text{minimize}} && \gamma(1-z) + (1-\gamma)(1-\hat{x}), \\
& \text{subject to} && \gamma + (1-\gamma)\textsf{h}(\hat{x}) = 1 - c.
\end{aligned}
\end{equation}
has the optimal solutions $\gamma(z) = (1-c-\textsf{h}(x(z))) / (1-\textsf{h}(x(z)))$ and $x(z)$ satisfying the fixed-point equation $x(z) = (1-x(z))^{\frac{z-1}{z+1}} - 1$.
\end{lemma}
\begin{proof}
See Appendix \ref{app-opt} for a complete proof.
\end{proof}

The next theorem characterizes the exact formula of $\underline{\Lambda}(z)$.
\begin{theorem}\label{theorem-underline-Lambda}
Consider the channel family $\{\textrm{BMS}(c)\}$. Then,
\begin{equation*}
\underline{\Lambda}(z) = \begin{cases}
2\epsilon_{\textsf{bsc}}& \text{if } 0 \leq z < z_{\textsf{bsc}},\\
\cfrac{1-c-\textsf{h}(x(z))}{1-\textsf{h}(x(z))}(1-z) \\
\quad\quad+ \cfrac{c}{1-\textsf{h}(x(z))}(1-x(z)) & \text{if } z_{\textsf{bsc}} \leq z < 1,\\
0 & \text{if } z=1,
\end{cases}
\end{equation*}
where $x(z)$ is the solution of $x = (1-x)^{\frac{z-1}{z+1}} - 1$ and
\begin{equation}\label{equ-z-bsc}
z_{\textsf{bsc}} \triangleq \cfrac{\log_2\left(4\epsilon_{\textsf{bsc}}\bar{\epsilon}_{\textsf{bsc}}\right)}{\log_2\left(\epsilon_{\textsf{bsc}} / \bar{\epsilon}_{\textsf{bsc}}\right)}.
\end{equation}
\end{theorem}
\begin{proof}
First, the equation $x(z) = (1-x(z))^{\frac{z-1}{z+1}} - 1$ for $z \in [0,1)$ is equivalent to
\begin{equation}\label{equ-zx}
z(x) = \cfrac{\log_2(1-x) + \log_2(1+x)}{\log_2(1-x) - \log_2(1+x)},
\end{equation}
where $\lim_{x \to 1} z(x) = 1$. Plugging $x=1-2\epsilon_{\textsf{bsc}}$ into (\ref{equ-zx}) yields (\ref{equ-z-bsc}). It is then clear that, for $z \in [0,z_{\textsf{bsc}})$,
\begin{equation*}
\underline{\Lambda}(z) = \min \{\Lambda_{\textsf{bsc}}(z), \Lambda_{\textsf{bec}}(z)\} = 2\epsilon_{\textsf{bsc}}.
\end{equation*}
Note that, for $z=z_{\textsf{bsc}}$, we have $\Lambda_{\textsf{bsc}}(z) = \Lambda_{\textsf{opt}}(z)$. For $z \in [z_{\textsf{bsc}},1)$, the monotonically non-increasing property of the $\Lambda$ function shows that $\Lambda_{\textsf{opt}}(z) \leq \Lambda_{\textsf{bsc}}(z)$. Moreover, when $x(z) \to 1$ as $z \to 1$, the continuity of $\textsf{h}(\cdot)$ shows that, the negative of the slope of $\Lambda_{\textsf{opt}}(z)$ is equal to
\begin{equation*}
\lim_{x \to 1} \cfrac{1-c-\textsf{h}(x(z))}{1-\textsf{h}(x(z))} = 1-c,
\end{equation*}
which is equal to the negative of the slope of $\Lambda_{\textsf{bec}}(z)$. Since the slope of $\Lambda_{\textsf{opt}}(z)$ at $z_{\textsf{bsc}}$ is 0 and $\Lambda_{\textsf{bec}}(z)$ has a constant slope, it follows from the convex-$\cap$ property of the $\Lambda$ function that $\Lambda_{\textsf{opt}}(z) \leq \Lambda_{\textsf{bec}}(z)$, for $z \in [z_{\textsf{bsc}},1)$. Consequently, when $z \in [z_{\textsf{bsc}},1)$, we have
\begin{align*}
\underline{\Lambda}(z) &= \gamma(z) (1-z) + \left(1-\gamma(z)\right) (1-x(z)) \\
&= \cfrac{1-c-\textsf{h}(x(z))}{1-\textsf{h}(x(z))}(1-z) + \cfrac{c}{1-\textsf{h}(x(z))}(1-x(z)).
\end{align*}
Trivially, when $z=1$, $\underline{\Lambda}(z) = 0$, completing the proof.
\end{proof}
Now denote by $\Lambda_{*}(z)$ the $\Lambda$ function of the least upgraded channel with respect to the whole channel family $\{\textrm{BMS}(c)\}$. Then, it is not difficult to see the following.
\begin{corollary}
$\Lambda_{*}(z) = \underline{\Lambda}(z)$, for $z \in [0,1]$.
\end{corollary}

Denote by $\mathrm{C}_{*}$ the capacity of the least upgraded channel. Expressing $z$ in terms of $x$, inserting the formula for $\Lambda_{*}(z(x))$ into (\ref{equ-H-alt}), and integrating over $x$ gives $\mathrm{C}_{*}$. However, there is no simple formula in the closed form for $\mathrm{C}_{*}$. Instead, we compute it numerically in Section \ref{sec-simul}.

\section{Simulation Results}\label{sec-simul}
In this section, we provide several simulation results regarding channel degradation and upgradation.
From Fig.~\ref{fig-lambda}, one can see that the difference between $\overline{\Lambda}(z)$ and $\Lambda^{*}(z)$ is on the interval $[0,1-2\epsilon_{\textsf{bsc}})$. On this region, $\overline{\Lambda}(z)$ is convex-$\cup$ while $\Lambda^{*}(z)$ is linear and thus convex-$\cap$.
\begin{figure}[!htb]
\centering
\subfigure{\includegraphics[width=4cm]{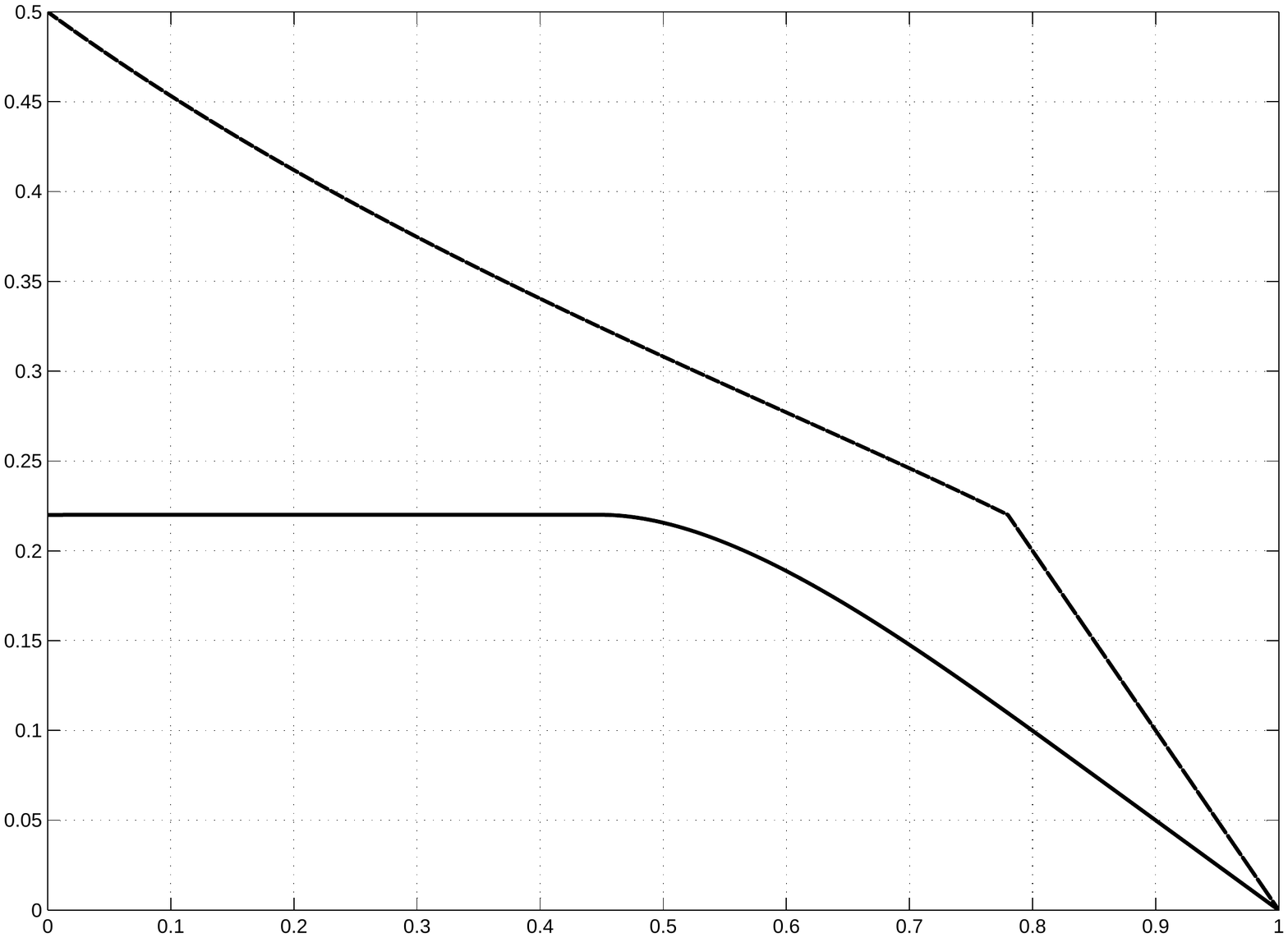}}
\subfigure{\includegraphics[width=4cm]{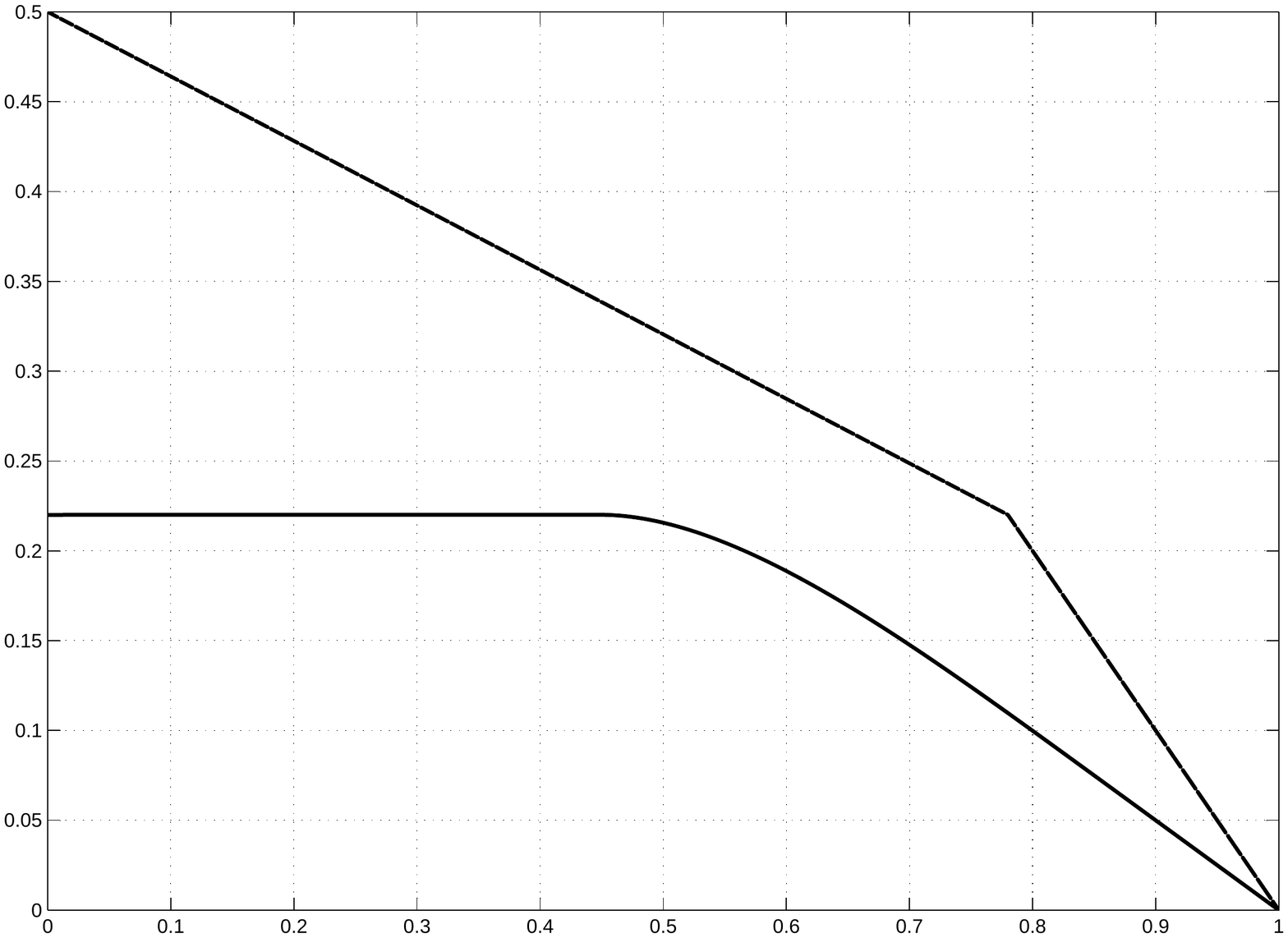}}
\caption{Left: The point-wise maximum $\overline{\Lambda}(z)$ (dashed line) and the point-wise minimum $\underline{\Lambda}(z)$ (solid line) of $\Lambda_a(z)$, for all $a \in \{\textrm{BMS}(0.5)\}$. Right: The $\Lambda$ functions of the least degraded channel $\Lambda^{*}(z)$ (dashed line) and the least upgraded channel $\Lambda_{*}(z)$ (solid line), respectively.}
\label{fig-lambda}
\end{figure}

Theorem \ref{theorem-overline-Lambda} and Theorem
\ref{theorem-underline-Lambda} suggest that channels with very few
probability masses always achieve the point-wise maximum or point-wise
minimum. Thus, we randomly generate $5,000$ BMS channels of
capacity $0.5$ having $2$ probability masses, and $5,000$ BMS channels of
capacity $0.5$ having $3$ probability masses.
Fig.~\ref{fig-channel-random} depicts the simulation results.
\begin{figure}[!htb]
\centering
\includegraphics[width=7cm]{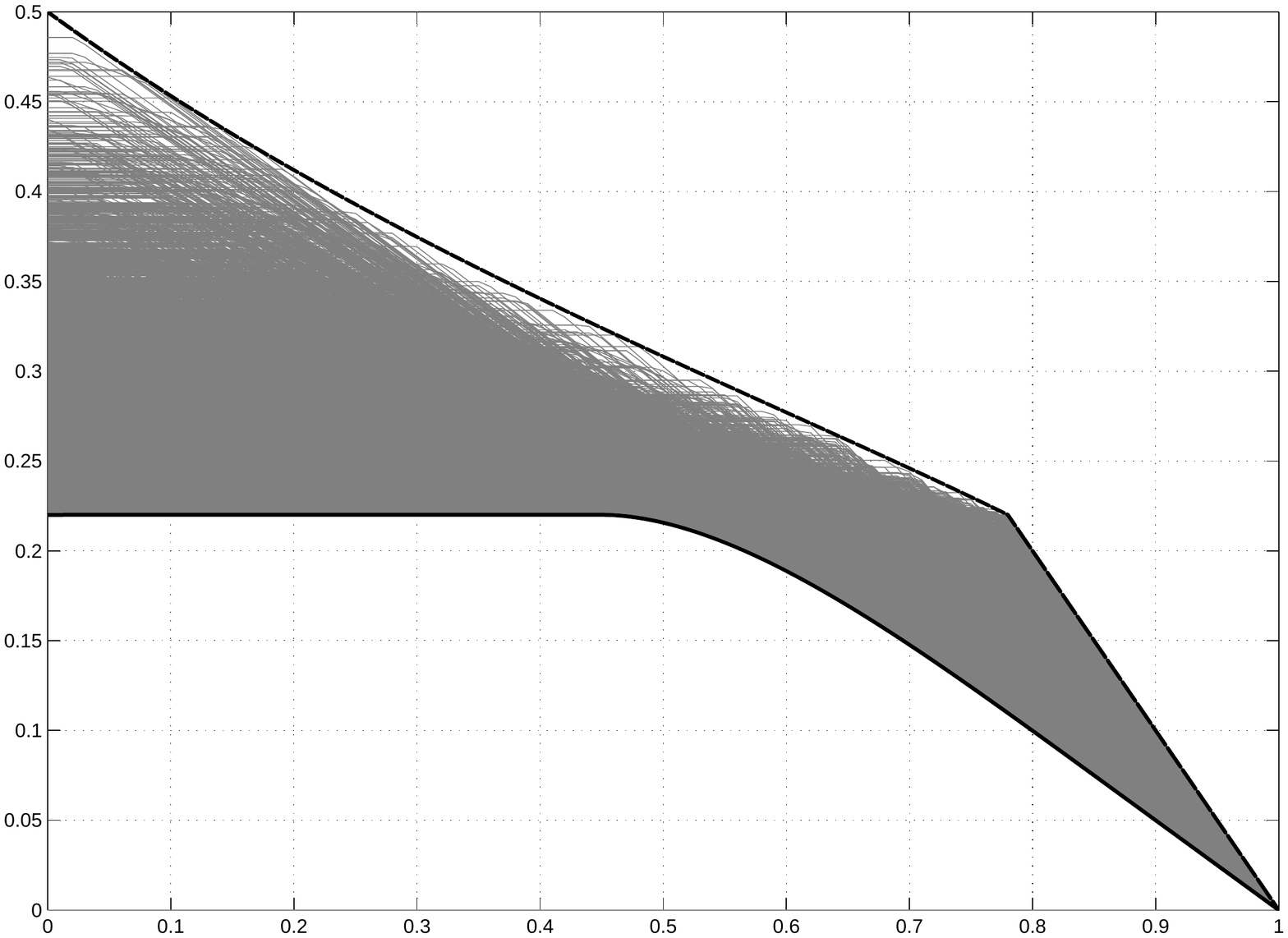}
\caption{The $\Lambda$ functions of randomly generated BMS channels (grey lines) from $\{\textrm{BMS}(0.5)\}$, and the point-wise maximum $\overline{\Lambda}(z)$ (dashed line) and minimum $\underline{\Lambda}(z)$ (solid line) of $\Lambda_a(z)$, for all $a \in \{\textrm{BMS}(0.5)\}$.}
\label{fig-channel-random}
\end{figure}

Other interesting quantities are
the gap between $c$ and the capacity of the least degraded channel, call it $d^{\textrm{gap}}$,
and the gap between the capacity of the least upgraded channel and $c$, call it $u^{\textrm{gap}}$.
See Table \ref{table-capacity-gap} for details.
\begin{table}[!htb]
\centering
\begin{tabular}{ccc}
$c$ & $d^{\textrm{gap}}$ & $u^{\textrm{gap}}$ \\
\toprule
$0.1$ & $0.0728$ & $0.0686$ \\
$0.2$ & $0.1222$ & $0.1009$ \\
$0.3$ & $0.1552$ & $0.1185$ \\
$0.4$ & $0.1739$ & $0.1257$ \\
$0.5$ & $0.1795$ & $0.1243$ \\
$0.6$ & $0.1721$ & $0.1155$ \\
$0.7$ & $0.1516$ & $0.0995$ \\
$0.8$ & $0.1175$ & $0.0762$ \\
$0.9$ & $0.0684$ & $0.0446$ \\
\end{tabular}
\caption{The gap between $c$ and $\mathrm{C}^{*}$ is $d^{\textrm{gap}} = c - \mathrm{C}^{*}$, and the gap between $\mathrm{C}_{*}$ and $c$ is $u^{\textrm{gap}} = \mathrm{C}_{*} - c$.}
\label{table-capacity-gap}
\end{table}

Fig.~\ref{fig-channel-capacity} depicts the capacity $\mathrm{C}^{*}$ ($\mathrm{C}_{*}$, respectively)
of the least degraded channel (resp. the least upgraded channel)
as a function of $c$, where $c$ ranges from $0.001$ to $0.999$ with a step size of $0.001$.
We then compute the maximum value of $d^{\textrm{gap}}$ being approximately $0.1795$ which corresponds to $c=0.4940$;
while the maximum value of $u^{\textrm{gap}}$ is approximately $0.1261$ which corresponds to $c=0.4310$.
\begin{figure}[!htb]
\centering
\includegraphics[width=7cm]{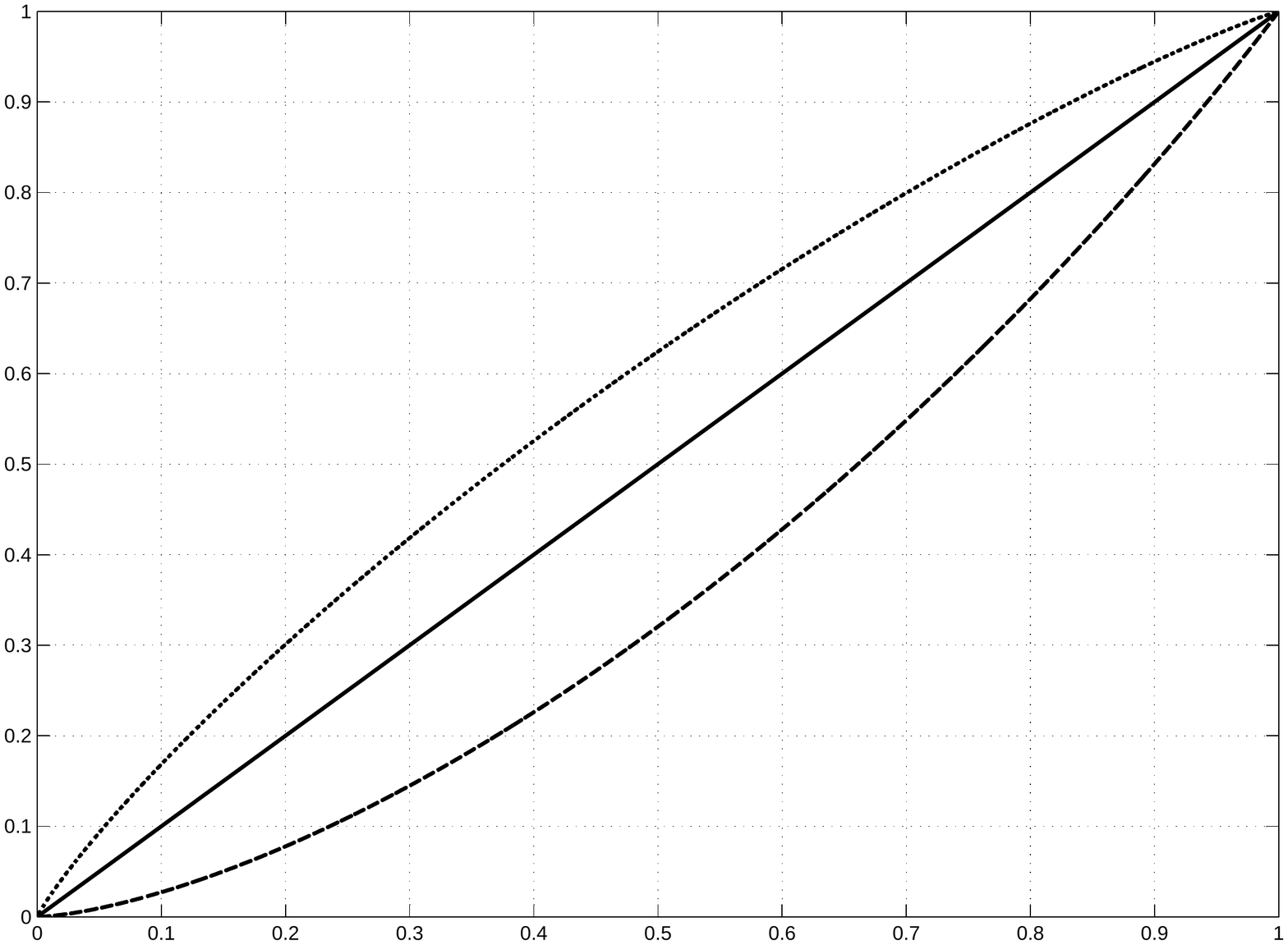}
\caption{The channel capacity of the least degraded channel (dashed line) as a function of the underlying channel capacity $c$, $0 < c < 1$; the channel capacity of the least upgraded channel (dotted line) as a function of the underlying channel capacity $c$, and the underlying channel capacity (solid line) as a function of itself.}
\label{fig-channel-capacity}
\end{figure}

\section{Conclusion and Open Problems}
The least degraded channel and the least upgraded channel with
respect to the family of BMS channels of fixed capacity were
introduced in this paper. Also, their characterizations and capacity
formulae were derived. An interesting open question is to consider
other families of channels. E.g., consider the family of channels
which are the result of taken the family of channels of fixed
capacity $c$ and performing one polarization step to them.  We
reserve such questions for future work.

\section*{Acknowledgement}
This work was supported by grant No. 200021-125347 of the Swiss National Foundation.

\appendix
\subsection{Proof of Lemma \ref{lemma-optimization}}\label{app-opt}
\begin{proof}
Define $f(\gamma,\hat{x},\eta) \triangleq \gamma(1-z) + (1-\gamma)(1-\hat{x}) - \eta \left(\gamma + (1-\gamma)\textsf{h}(\hat{x}) - 1 + c\right)$. Then, taking the first-order derivatives with respect to $\hat{x}$, $\gamma$, and $\eta$ yields
\begin{align}
\label{equ-para-x}\cfrac{\partial}{\partial\hat{x}}f(\gamma,\hat{x},\eta) &= \gamma - 1 - \eta(1-\gamma)\textsf{h}'(\hat{x}),\\
\cfrac{\partial}{\partial\gamma}f(\gamma,\hat{x},\eta) &= 1-z-(1-x)-\eta+\eta \textsf{h}(\hat{x}),\\
\cfrac{\partial}{\partial\eta}f(\gamma,\hat{x},\eta) &= \gamma + (1-\gamma)\textsf{h}(\hat{x}) - 1 + c.
\end{align}
Now setting the first-order derivatives to be 0 gives
\begin{equation}\label{equ-para-gamma}
\gamma = \frac{1-c-\textsf{h}(\hat{x})}{1-\textsf{h}(\hat{x})} \quad \text{and} \quad \eta = \frac{\hat{x}-z}{1-\textsf{h}(\hat{x})}.
\end{equation}
Plugging (\ref{equ-para-gamma}) into (\ref{equ-para-x}) and making some manipulations give
\begin{equation*}
(1+z)\log_2(1+\hat{x}) + (1-z) \log_2(1-\hat{x}) = 0,
\end{equation*}
which is equivalent to
\begin{equation*}
\hat{x} = (1-\hat{x})^{\frac{z-1}{z+1}} - 1.
\end{equation*}
\end{proof}


\begin{thebibliography}{1}
\bibitem{ru08} T. Richardson and R. Urbanke, {\em Modern Coding Theory}. Cambridge, UK. Cambridge University Press, 2008.
\bibitem{arikan09} E. Ar{\i}kan, ``Channel polarization: A method for constructing capacity-achieving codes for symmetric binary-input memoryless channels,'' {\em IEEE Trans. Inform. Theory}, vol. 55, pp. 3051-3073, 2009.
\end{thebibliography}
\end{document}